\documentclass[9pt,twocolumn]{IEEEtran}
\usepackage{amsmath ,amssymb,euscript ,yfonts,psfrag,latexsym,dsfont,graphicx,bbm,color,amstext,wasysym,subfig,flushend,parskip}
\graphicspath{{./},{./figurespdf/}}

\begin{document}
\newtheorem{thm}{Theorem}
\newtheorem{cor}[thm]{Corollary}
\newtheorem{conj}[thm]{Conjecture}
\newtheorem{lemma}[thm]{Lemma}
\newtheorem{prop}[thm]{Proposition}
\newtheorem{problem}[thm]{Problem}
\newtheorem{remark}[thm]{Remark}
\newtheorem{defn}[thm]{Definition}
\newtheorem{ex}[thm]{Example}

\newcommand{\mR}{{\mathbb R}}
\newcommand{\D}{{\mathbb D}}
\newcommand{\E}{{\mathbb E}}
\newcommand{\mcN}{{\mathcal N}}
\newcommand{\mcR}{{\mathcal R}}
\newcommand{\diag}{\operatorname{diag}}
\newcommand{\tr}{\operatorname{trace}}
\newcommand{\ignore}[1]{}

\def\spacingset#1{\def\baselinestretch{#1}\small\normalsize}
\setlength{\parskip}{10pt}
\setlength{\parindent}{20pt}
\spacingset{1}

\definecolor{grey}{rgb}{0.6,0.6,0.6}
\definecolor{lightgray}{rgb}{0.97,.99,0.99}

\title{Optimal steering of a linear stochastic system\\
to a final probability distribution}

\author{Yongxin Chen, Tryphon Georgiou and Michele Pavon
\thanks{Y.\ Chen and T.T.\ Georgiou are with the Department of Electrical and Computer Engineering,
University of Minnesota, Minneapolis, Minnesota MN 55455, USA; {\sf\footnotesize email: \{chen2468,tryphon\}@umn.edu}}
\thanks{M.\ Pavon is with the Dipartimento di Matematica,
Universit\`a di Padova, via Trieste 63, 35121 Padova, Italy; {\sf\footnotesize email: pavon@math.unipd.it}}}
\markboth{\today}{}

\maketitle
\begin{abstract}
We consider the problem to steer a linear dynamical system with full state observation from an initial gaussian distribution in state-space to a final one with minimum energy control. The system is stochastically driven through the control channels; an example for such a system is that of an inertial particle experiencing random ``white noise'' forcing.
We show that a target probability distribution can always be achieved in finite time. The optimal control is given in state-feedback form and is computed explicitely by solving a pair of differential Lyapunov equations that are {\em coupled} through their boundary values.  This result, given its attractive algorithmic nature, appears to have several potential applications such as to active control of nanomechanical systems and molecular cooling. The problem to steer a diffusion process between end-point marginals has a long history (Schr\"odinger bridges) and therefore, the present case of steering a linear stochastic system constitutes a Schr\"odinger bridge for possibly degenerate diffusions. Our results, however, provide the first {\em implementable} form of the optimal control for a general Gauss-Markov process. Illustrative examples of the optimal evolution and control for inertial particles and a stochastic oscillator are provided.
A final result establishes directly the property of Schr\"{o}dinger bridges as the most likely random evolution between given marginals to the present context of linear stochastic systems.

\end{abstract}


\section{Introduction}
Active control of micro-mechanical systems has witnessed major advances in the past twenty years. At the atomic scale, control of quantum mechanical systems has also enormously increased its scope and effective laboratory implementation. We mention laser-driven molecular reactions, design of pulse sequences in NMR, adaptive quantum measurements and feedback control of optical systems. For a recent survey see \cite{AT} which is addressed to a control engineering audience.  Another important area where feedback control is playing an increasing role is cooling. Advances in nanotechnology permit nowadays to implement feedback
control actions on nanodevices \cite{BBP}. For instance, in surface
topography, the deflection of a cantilever is captured by a
photodetector that records the angle of reflection from a laser
beam focused on the mirrored surface on back side of the cantilever. Position feedback control is used to maintain the probe at a constant force or distance from the object surface. Position can also be differentiated allowing to apply a velocity dependent external force. A velocity dependent feedback control (VFC) has been implemented to reduce thermal noise of a cantilever in atomic force microscopy (AFM) \cite{LMC} and in dynamic force microscopy \cite{THO}. Another important area of application is polymer dynamics \cite{DE}.

Cooling is of interest for non-microscopic systems as well. For macroscopic mechanical resonators, for instance, cooling to ultralow temperatures is indispensable to investigate decoherence. In \cite{Vin,BCD}, a feedback cooling technique  on a ton-scale resonant-bar gravitational wave detector and the corresponding thermodynamics is described.  See \cite{PV} for a recent survey of cooling techniques for both meter-sized detectors and nanomechanical systems. In all cooling applications, the basic model is a stochastic oscillator \cite{N1} which is driven asymptotically to a desired non equilibrium steady state by means of feedback control. These diffusion mediated devices are sometimes called {\em Brownian motors} as work can be extracted from them \cite{R}. The issue of {\em motor efficiency} translates into an optimal control problem \cite{FH}. Indeed, this problem may be viewed as a special case of the theory of Schr\"{o}dinger bridges for diffusion processes \cite{W} where the time interval is infinite \cite{CGP}. The connection between these problems and the so called ``logarithmic transformation" of stochastic control of Fleming, Holland, Mitter  {\em et al.}, see e.g. \cite {F}, has been investigated for some time \cite{daipra,daiprapavon,PW,FH,FHS}.

In spite of this large body of work, the situation is far from satisfactory considering the challenges and opportunities offered by modern technology in controlling micro and macro mechanical systems. One drawback is that the basic theory of the Schr\"odinger bridges has been developed for non degenerate diffusions where the noise acts on all components of the state vector, whereas the stochastic oscillators of interests are degenerate diffusions in phase space. A much more serious problem is that the solution, excepting very special cases \cite{FH,FHS}, is in general not given in a form amenable to computations. Indeed, computing the optimal control requires solving a pair of partial differential equations coupled through their boundary values \cite{W}.

 The purpose of this paper is to partially remedy this situation. We provide what can be regarded as the first {\em computable} and {\em implementable} solution in the important case of a Gauss-Markov process (nonlinear stochastic oscillators are considered in \cite{CGP}). This case had been discussed in the discrete time setting in \cite{beghi}. However, the existence and an implementable form of the optimal control are missing in this paper and, moreover, the noise intensity is assumed to be nonsingular. Another related line of research, in the work by Robert Skelton and his co-workers \cite{skelton}, has been to assign the {\em asymptotic} closed-loop state-covariance with dynamic output feedback. In spite of the fact that control takes place over an infinite time interval, here too, computational aspects and conditions for ``assignability'' of steady state covariance are far from trivial.

In the present work we show that a linear dynamical system can be optimally steered from {\em any} initial Gaussian distribution for the initial state to {\em any} final one, over any finite interval $[0,T]$. The unique minimum-energy state-feedback control is explicitly constructed by solving two linear Lyapunov differential equations. These are {\em nonlinearly coupled} through boundary conditions at the two end points of the interval. However, we show that these boundary values can be expressed in {\em closed form} as (nonlinear) functions of the covariances for the initial and target Gaussian distributions.

The paper is structured as follows. The formulation of the main problem and the variational analysis that shows the form of the optimal control are given in Section \ref{sec:variational}. The existence and the explicit construction of the optimal control is given in Section \ref{sec:III}. Although the state process may be a degenerate diffusions (since, typically, the rank of the input matrix is typically less than the dimension of the state vector), the law of the controlled dynamics is closest in the relative entropy sense to that of the uncontrolled dynamics, just as in the theory of the Schr\"odinger bridges; this is shown in Section \ref{sec:minentropy}. Finally, in Section \ref{sec:examples} we present two illustrative examples. The first one is on {\em inertial} particles experiencing random (white) acceleration, and the second, on active damping of an oscillator driven by Nyquist-Johnson thermal noise.

\section{Problem formulation and variational analysis}\label{sec:variational}
Consider a ``prior" evolution given by the vector Gauss-Markov process $\{x(t) \mid 0\le t\le T\}$ satisfying the $n$-dimensional linear stochastic differential equation
\begin{align}\label{prior}
dx(t)=&A(t)x(t)dt+B(t)dw(t)\\\nonumber &\mbox{with } x(0)=\xi \mbox{ a.s. }
\end{align}
and $\xi$ an $n$-dimensional random vector independent of $\{w(t)\mid 0\le t\le T\}$ with density
\begin{equation}\label{initial}\rho_0(x)=(2\pi)^{-n/2}\det (\Sigma_0)^{-1/2}\exp\left(-\frac{1}{2}x'\Sigma_0^{-1}x\right).
\end{equation}
Throughout,
 $\{w(t) \mid 0\leq t\leq T\}$ is a standard, $m$-dimensional Wiener process and   $A(\cdot)$ and $B(\cdot)$ are continuous matrix functions taking values in $\mR^{n\times n}$ and $\mR^{n\times m}$, respectively.
Consider also the controlled evolution
\begin{align}\label{controlled}dx^u(t)&=A(t)x^u(t)dt+B(t)u(t)+B(t)dw(t),\\\nonumber &\quad x^u(0)=\xi \mbox{ a.s.}
\end{align}
and a ``target'' end-point distribution\begin{equation}\label{final}\rho_T(x)=(2\pi)^{-n/2}\det (\Sigma_T)^{-1/2}\exp\left(-\frac{1}{2}x'\Sigma_T^{-1}x\right),
\end{equation}
which is Gaussian with zero mean with covariance
$\Sigma_T>0$, we
let $\mathcal U$ be the family of {\em adapted}, {\em finite-energy} control functions such that (\ref{controlled}) has a strong solution and $x^u(T)$ is distributed according to \eqref{final}.
More precisely, $u\in\mathcal U$ is such that  $u(t)$ only depends on $t$ and on $\{x^u(s); 0\le s\le t\}$ for each $t\in [0,T]$, satisfies
$$\E\left\{\int_0^Tu(t)' u(t) \,dt\right\}<\infty,
$$
and effects $x^u(T)$ to be distributed according to \eqref{final}.
The family $\mathcal U$ represents  {\em admissible} control inputs which achieve the desired probability density transfer from $\rho_0$ to $\rho_T$. Thence we formulate the following {\em Schr\"{o}dinger Bridge Problem:}

\begin{problem}\label{formalization} Determine whether ${\mathcal U}$ is non-empty and if so, determine
$
u^*:= {\rm argmin}_{u\in \mathcal U} \,J(u)$
where
\[
J(u):=\E\left\{\int_0^Tu(t)' u(t) \,dt\right\}.
\]
\end{problem}

In the next section we will prove that a minimizing control $u^*$ always exists. The stochastic process
$
\{x^*(t)=x^{u^*}(t) \mid 0\le t\le 1\}
$
will be referred to as the {\em Schr\"{o}dinger bridge from $\rho_0$ to $\rho_1$ over the prior $\{x(t)=x^0(t) \mid 0\le t\le 1\}$}.

Notice that in the ``controlled'' equation (\ref{controlled}) the control variables $u(t)$ act through the same input ``channels'' which are subject to noise, i.e., both $u(t)$ and $dw(t)$ affect the state through the same $B(\cdot)$ matrix. The theory that follows can accordingly be relaxed to the case where the control has more ``authority'' (i.e., the range of the corresponding $B$-matrix contains the range of the $B$-matrix for the noise). It will be of interest to study in detail the case where the control authority is less than that of the stochastic noise.

In the remaining of the section we identify a candidate structure for the optimal controls and reduce the problem to an algebraic condition involving two differential Lyapunov equations that are nonlinearly coupled through split boundary conditions.

Let us start by observing that this problem resembles a standard stochastic linear quadratic regulator problem except for the boundary conditions. The usual variational analysis can in fact be carried out, up to a point, namely the expression for the optimal control, in a similar fashion. Of the several ways in which the form of the optimal control can be obtained, we choose a most familiar one, namely the so-called ``completion of squares"\footnote{Although it might be the most familiar to control engineers, the completion of the square argument for stochastic linear quadratic regulator control is not the most elementary. Indeed, a derivation which does not employ It${\rm \bar o}$'s rule was presented in \cite{KP92}.}.  Let $\{\Pi(t) \mid 0\le t\le T\}$ be a differentiable function taking values in the set of symmetric, $n\times n$ matrices. Observe that Problem \ref{formalization} is equivalent to minimizing over $\mathcal U$ the modified index
\begin{align}\label{modind}
\tilde J(u)&=\E\left\{\int_0^Tu(t)' u(t) \,dt\right.\\
&\left.+x(T)'\Pi(T)x(T)-x(0)'\Pi(0)x(0)\right\}.\nonumber
\end{align}
Indeed, as the two end-point marginals densities $\rho_0$ and $\rho_T$ are fixed when $u$ varies in $\mathcal U$, the two boundary terms are constant over $\mathcal U$. We can now rewrite $\tilde J(u)$ as follows
$$\tilde J(u)=\E\left\{\int_0^Tu(t)' u(t) \,dt +\int_0^Td\left(x(t)'\Pi(t)x(t)\right)\right\}.
$$
Assuming that on $[0,T]$ $\Pi(t)$ satisfies the matrix Riccati equation
\begin{equation}\label{R1}
\dot{\Pi}(t)=-A(t)'\Pi(t)-\Pi(t)A(t)+\Pi(t)B(t)B(t)'\Pi(t),
\end{equation}
a standard argument using It$\rm \bar o$'s rule (e.g., see \cite{FR}) shows that
\begin{align*}
\tilde J(u)&=\E\left\{\int_0^T\|u(t)+B(t)'\Pi(t)x(t)\|^2 \,dt\right.\\&\left. + \int_0^T\frac{1}{2}\tr\left(\Pi(t)B(t)B(t)'\right)dt\right\}.
\end{align*}
Observe that the second integral is finite and invariant over $\mathcal U$. Hence, we obtain a candidate for the optimal control in the familiar form
\begin{equation}\label{optcontr}
u^*(t)=-B(t)'\Pi(t)x(t).
\end{equation}
Such a choice of control will be possible
provided we can find a solution $\Pi(t)$ of (\ref{R1}) such that the process
\begin{align}\label{optevolution}
dx^*(t)&=\left(A(t)-B(t)B(t)'\Pi(t)\right)x^*(t)dt+B(t)dw(t),\\
&\quad \mbox{with }x^*(0)=\xi \mbox{ a.s. }\nonumber
\end{align}
leads to $x^*(T)$ with density $\rho_T$. If this is indeed possible, then we have solved Problem \ref{formalization}.
It is important to observe that the optimal control, if it exists, is in a {\em state feedback} form. Consequently, the new optimal evolution is a {\em Gauss-Markov} process just as the prior evolution.

Finding the solution of the Riccati equation which achieves the density transfer is nontrivial. In the classical linear quadratic regulator theory, the terminal cost of the index would provide the boundary value $\Pi(T)$ for (\ref{R1}). However, here there is no boundary value and the two analyses sharply bifurcate. Therefore, we need to resort to something quite different as we have information concerning both initial and final densities, namely $\Sigma_0$ and $\Sigma_T$.

Let $\Sigma(t):=\E\left\{x^*(t)x^*(t)'\right\}$ be the state covariance of the sought optimal evolution. From (\ref{optevolution}) we have that $\Sigma(t)$ satisfies
\begin{align}\nonumber\dot{\Sigma}(t)&= \left(A(t)-B(t)B(t)'\Pi(t)\right)\Sigma(t)\\&\hspace*{-5pt}+\Sigma(t)\left(A(t)-B(t)B(t)'\Pi(t)\right)'+B(t)B(t)'.\label{sigmaevol}
\end{align}
It must also satisfy the two boundary conditions
\begin{equation}\label{BND}
\Sigma(0)=\Sigma_0, \quad \Sigma(T)=\Sigma_T
\end{equation}
and, provided $(A(t),B(t))$ is controllable (see Section \ref{sec:III}), $\Sigma(t)$ is positive definite on $[0,T]$. Thus, we seek a solution pair $(\Pi(t),\Sigma(t))$ of the {\em coupled} system of these two equations \eqref{R1} and \eqref{sigmaevol} with split boundary conditions \eqref{BND}.

Interestingly, if we
define the new matrix-valued function
$${\rm H}(t):=\Sigma(t)^{-1}-\Pi(t),
$$
then a direct calculation using (\ref{sigmaevol}) and (\ref{R1}) shows that ${\rm H}(t)$ satisfies the homogeneous Riccati equation
\begin{equation}\label{R2}
\dot{\rm H}(t)=-A(t)'{\rm H}(t)-{\rm H}(t)A(t)-{\rm H}(t)B(t)B(t)'{\rm H}(t).
\end{equation}
This equation is dual to \eqref{R1} and the system of the two coupled {matrix} equations \eqref{R1} and \eqref{sigmaevol} can be replaced by
\eqref{R1} and \eqref{R2}. The new system is {\em decoupled}, except for the coupling through their boundary conditions
\begin{subequations}\label{bndphi}
\begin{eqnarray}
\Sigma_0^{-1}&=&\Pi(0)+{\rm H}(0)\\
\Sigma_T^{-1}&=&\Pi(T)+{\rm H}(T).
\end{eqnarray}
\end{subequations}
These boundary conditions \eqref{bndphi} are sufficient for meeting the two end-point marginals $\rho_0$ and $\rho_T$ provided of course that $\Pi(t)$ remains finite. {We have therefore established the following result.

\begin{prop}\label{systemRiccati} Suppose $\Pi(t)$ and ${\rm H}(t)$  satisfy equations \eqref{R1}-\eqref{R2} on $[0,T]$ with boundary conditions (\ref{bndphi}). Then the feedback control $u^*$ given in (\ref{optcontr}) is optimal for Problem \ref{formalization} and the optimal evolution of the Schr\"{o}dinger bridge is given by (\ref{optevolution}).
\end{prop}}
{Since (\ref{R1}) and (\ref{R2}) are homogeneous, they always admit the zero solution.} The case $\Pi(t)\equiv 0$ corresponds to the situation where the prior evolution satisfies the boundary marginals conditions and, in that case, ${\rm H}(t)^{-1}$ is simply the prior state covariance.

Thus, Problem \ref{formalization} reduces to the atypical situation of {\em two} Riccati equations (\ref{R1}) and (\ref{R2}) coupled through their boundary values. This might still at first glance appear to be a formidable problem. However, (\ref{R1})-(\ref{R2}) are homogeneous and, {as far as their non singular solutions}, they reduce to {\em linear} differential Lyapunov equations. {The latter, however,} are still coupled through their boundary values in a {\em nonlinear} way. {Indeed,} suppose $\Pi(t)$ exists on the time interval $[0,T]$ and is invertible. Then $Q(t)=\Pi(t)^{-1}$ satisfies the linear equation
\begin{subequations}\label{coupledLyapunov}
\begin{equation}\label{lin1}
\dot{Q}(t)=A(t)Q(t)+Q(t)A(t)'-B(t)B(t)'.
\end{equation}
Likewise, if ${\rm H}(t)$ exists on the time interval $[0,T]$ and is invertible, $P(t)={\rm H}(t)^{-1}$ satisfies the linear equation
\begin{equation}\label{lin2}
\dot{P}(t)=A(t)P(t)+P(t)A(t)'+B(t)B(t)'.
\end{equation}
\end{subequations}
The boundary conditions (\ref{bndphi}) for this new pair $(P(t),Q(t))$ now read
\begin{subequations}\label{bndphi2}
\begin{eqnarray}
\Sigma_0^{-1}&=&P(0)^{-1}+Q(0)^{-1}\\
\Sigma_T^{-1}&=&P(T)^{-1}+Q(T)^{-1}.
\end{eqnarray}
\end{subequations}
Conversely, if $Q(t)$ solves \eqref{lin1} and is nonsingular on $[0,T]$, then $Q(t)^{-1}$ is a solution of (\ref{R1}), and similarly for $P(t)$. {We record the following immediate consequence of Proposition \ref{systemRiccati}.
\begin{cor}\label{corollary} Suppose $P(t)$ and $Q(t)$ are nonsingular on $[0,T]$ and satisfy the equations (\ref{lin1}-\ref{lin2}) with boundary conditions (\ref{bndphi}). Then the feedback control
\begin{equation}\label{optcontr2}
u^*(t)=-B(t)'Q(t)^{-1}x(t).
\end{equation}
is optimal for Problem \ref{formalization}. The evolution of the  optimal Gauss-Markov process is given by
\begin{align}\label{optevolution2}
dx^*(t)&=\left(A(t)-B(t)B(t)'Q(t)^{-1}\right)x^*(t)dt+B(t)dw(t),\\
&\quad \mbox{with }x^*(0)=\xi \mbox{ a.s. }\nonumber
\end{align}
\end{cor}
Thus, the system \eqref{coupledLyapunov}-\eqref{bndphi2} appears as the {\em bottleneck} of the Schr\"{o}dinger bridge problem. }In the next section, we prove that in fact \eqref{coupledLyapunov}  always has solution $(P(t),Q(t))$, with both $P(t)$ and $Q(t)$ nonsingular on $[0,T]$, that satisfies \eqref{bndphi2} and that this solution is unique.

\section{Existence and uniqueness of optimal control for the linear Gaussian bridge}\label{sec:III}
We  assume throughout that the system \eqref{prior} (or equivalently the pair $(A(t),B(t))$) is controllable in the sense that the reachability gramian
\[M(t_1,t_0):=\int_{t_0}^{t_1} \Phi(t_1,\tau)B(\tau)B(\tau)^\prime\Phi(t_1,\tau)^\prime d\tau,\]
is nonsingular for all $t_0<t_1$ (with $t_0,t_1\in[0,T]$). As usual,
 $\Phi(t,s)$ denotes the state-transition matrix of \eqref{prior} determined via
\begin{eqnarray*}
\frac{\partial}{\partial t}\Phi(t,s)&=&A(t)\Phi(t,s) \mbox{ and }\Phi(t,t)=I,\end{eqnarray*}
and this is nonsingular for all $t,s\in[0,T]$.
It is worth noting that for $t_1>0$ the reachability grammian
$
M(t_1,0)=P(t_1)>0
$
satisfies the differential Lyapunov equation \eqref{lin2}
with $P(0)=0$.
The controllability grammian
\[
N(t_1,t_0):=\int_{t_0}^{t_1} \Phi(t_0,\tau)B(\tau)B(\tau)^\prime\Phi(t_0,\tau)^\prime d\tau,
\]
is necessarily also nonsingular for all $t_0<t_1$ ($t_0,t_1\in[0,T]$) and if, we similarly set
$Q(t_0)=N(T,t_0)$, then $Q(t)$ satisfies \eqref{lin1} with $Q(T)=0$.

However, as suggested in the previous section, we need to consider solutions $P(\cdot)$, $Q(\cdot)$ of these two differential Lyapunov equations \eqref{coupledLyapunov}
that satisfy boundary conditions that are coupled through \eqref{bndphi2}.
In general, $P(t)$ and $Q(t)$ do not need to be sign definite, but in order for
\begin{equation}\label{eq:Sinv}
\Sigma(t)^{-1}=P(t)^{-1}+Q(t)^{-1}.
\end{equation}
to qualify as a covariance of the controlled process \eqref{controlled} $P(t)$ and $Q(t)$ need to be invertible. This condition is also sufficient and $\Sigma(t)$ satisfies the corresponding differential Lyapunov equation for the covariance of the controlled process \eqref{optevolution2}
\begin{equation}\label{eq:lyapAQ}
\dot \Sigma(t) =
A_Q(t)\Sigma(t)+
\Sigma(t)A_Q(t)^\prime + B(t)B(t)^\prime
\end{equation}
with
\begin{equation}\label{eq:AQ}
A_Q(t):=(A(t)-B(t)B(t)^\prime Q(t)^{-1}).
\end{equation}

Next, we present our main technical result on the existence and uniqueness of an admissible pair $(P_-(t),Q_-(t))$ of solutions to (\ref{coupledLyapunov})-(\ref{bndphi2})  that are invertible on $[0,T]$. Interstingly, there is always a second solution  $(P_+(t),Q_+(t))$ to the nonlinear  problem (\ref{coupledLyapunov})-(\ref{bndphi2})  which is not admissible as it fails to be invertible on $[0,T]$.

\begin{prop}\label{prop:schpair} Consider $\Sigma_0,\Sigma_T>0$ and a controllable pair $(A(t),B(t))$ as before. The system of the two differential Lyapunov equations \eqref{coupledLyapunov} has two sets of solutions $(P_\pm(\cdot),Q_\pm(\cdot))$ over $[0,T]$ that simultaneously satisfy the coupling boundary conditions (\ref{bndphi2})
These two solutions are specified by
\begin{eqnarray*}
Q_\pm(0)
&=&N(T,0)^{1/2}S_0^{1/2}\left(S_0+\frac12 I\pm \left(S_0^{1/2}
S_TS_0^{1/2}\right.\right.\\
&&\left.\left.+\frac{1}{4}I\right)^{1/2}\right)^{-1}S_0^{1/2}N(T,0)^{1/2},\\
P_\pm(0)&=&\left(\Sigma_0^{-1}-Q_\pm(0)^{-1}\right)^{-1}
\end{eqnarray*}
and the two differential equations \eqref{coupledLyapunov}, where
\begin{align*}
S_0&=N(T,0)^{-1/2}\Sigma_0N(T,0)^{-1/2},\\
S_T&=N(T,0)^{-1/2}\Phi(0,T)\Sigma_T\Phi(0,T)N(T,0)^{-1/2}.
\end{align*}
The two pairs $(P_\pm(t),Q_\pm(t))$ with subscript $-$ and $+$, respectively, are distinguished by the following:
\begin{itemize}
\item[i)]
$Q_-(t)$ and $P_-(t)$ are both nonsingular on $[0,T]$, whereas
\item[ii)] $Q_+(t)$ and $P_+(t)$ become singular for some $t\in[0,T]$, possibly not for the same value of $t$.
\end{itemize}
\end{prop}
\begin{proof}
Apply the time-varying change of coordinates
\[
\xi(t)=N(T,0)^{-1/2}\Phi(0,t)x(t).
\]
Then, in this new coordinates the dynamical system \eqref{prior} becomes
\[
d\xi(t)=\underbrace{N(T,0)^{-1/2}\Phi(0,t)B(t)}_{B_{\rm new}(t)} dw(t).
\]
We will prove the statement in this new set of coordinates for the state and revert back to the one at the end. Accordingly,
\[
\dot P_{\rm new}(t)=B_{\rm new}(t)B_{\rm new}(t)^\prime,
\]
\[
\dot Q_{\rm new}(t)=-B_{\rm new}(t)B_{\rm new}(t)^\prime,
\]
along with $M_{\rm new}(T,0)=N_{\rm new}(T,0)=I$ and
\begin{subequations}
\begin{equation}\label{sigma0new}
\Sigma_{0,\rm new}=N(T,0)^{-1/2}\Sigma_0N(T,0)^{-1/2},
\end{equation}
while
\begin{equation}\label{sigmatnew}
\Sigma_{T,\rm new}=N(T,0)^{-1/2}\Phi(0,T)\Sigma_T\Phi(0,T)^\prime N(T,0)^{-1/2}.
\end{equation}
\end{subequations}
{
The relation between $Q_{\rm new}(t)$ and $Q(t)$ is given by
\[
    Q_{\rm new}(t)=N(T,0)^{-1/2}\Phi(0,t)Q(t)\Phi(0,t)'N(T,0)^{-1/2}.
\]
This can be seen by taking the derivative of both sides
\begin{eqnarray*}
    \dot Q_{\rm new}(t)&=& -N(T,0)^{-1/2}\Phi(0,t)A(t)Q(t)\Phi(0,t)'N(T,0)^{-1/2}\\
                        &&-N(T,0)^{-1/2}\Phi(0,t)Q(t)A(t)'\Phi(0,t)'N(T,0)^{-1/2}\\
                        && +N(T,0)^{-1/2}\Phi(0,t)\dot{Q}(t)\Phi(0,t)'N(T,0)^{-1/2}\\
                        &=& -N(T,0)^{-1/2}\Phi(0,t)B(t)B(t)'\Phi(0,t)'N(T,0)^{-1/2}\\
                        &=& -B_{\rm new}(t)B_{\rm new}(t)^\prime.
\end{eqnarray*}
}
In the next paragraph, for notational convenience, we drop the subscript ``new'' and prove the statement assuming that $A(t)=0$ as well as
 $N(T,0)=I$. We will return to the notation that distinguishes the two sets of coordinates with the subscript ``new'' and relate back to the original ones at the end of the proof.

Since $A(t)=0$, then $\Phi(t,x)=I$ for all $s,t\in[0,T]$. Further, $M(T,0)=N(T,0)=I$. Thus,
\begin{eqnarray*}
P(T)&=&P(0)+I\\
Q(T)&=&Q(0)-I.
\end{eqnarray*}
Substituting in (\ref{bndphi2}), we
obtain that
\begin{eqnarray*}
Q(0)^{-1}+P(0)^{-1}&=&\Sigma_0^{-1}\\
(Q(0)-I)^{-1}+(P(0)+I)^{-1}&=&\Sigma_T^{-1}.
\end{eqnarray*}
Solving the first for $P(0)$ as a function of $Q(0)$ and substituting in the second, we have
\begin{eqnarray*}
\Sigma_T^{-1}&=&((\Sigma_0^{-1}-Q(0)^{-1})^{-1}+I)^{-1}+(Q(0)-I)^{-1}\\
&=&
((\Sigma_0^{-1}-Q(0)^{-1})^{-1}+I)^{-1}\\
&&\hspace*{10pt}\times(
Q(0)+(\Sigma_0^{-1}-Q(0)^{-1})^{-1})
)(Q(0)-I)^{-1}\\
&=&
((\Sigma_0^{-1}-Q(0)^{-1})^{-1}+I)^{-1}\\
&&\hspace*{10pt}\times
(\Sigma_0^{-1}-Q(0)^{-1})^{-1}\Sigma_0^{-1}Q(0)(Q(0)-I)^{-1}\\
&=&(\Sigma_0^{-1}+I-Q(0)^{-1})^{-1}\Sigma_0^{-1}(I-Q(0)^{-1})^{-1}
\end{eqnarray*}
which after inversion leads to
\[
(I-Q(0)^{-1})\Sigma_0(I-Q(0)^{-1})+(I-Q(0)^{-1})=\Sigma_T.
\]
This is a quadratic expression and has two Hermitian solutions
\begin{equation}\label{eq:finite}
I-Q(0)^{-1}=\Sigma_0^{-1/2}
\left(-\frac12 I \mp \left(\Sigma_0^{1/2}\Sigma_T\Sigma_0^{1/2}+\frac{1}{4}I\right)^{1/2}\right)
\Sigma_0^{-1/2}.
\end{equation}

This gives that
\[
Q(0)=\Sigma_0^{1/2}\left(\Sigma_0+\frac12 I\pm \left(\Sigma_0^{1/2}
\Sigma_T\Sigma_0^{1/2}+\frac{1}{4}I\right)^{1/2}\right)^{-1}\Sigma_0^{1/2}.
\]
To see that i) holds evaluate (in these simplified coordinates where there is no drift and $M(T,0)=I$)
\begin{eqnarray*}
Q_-(t)^{-1}&=&(Q_-(0)-M(t,0))^{-1}\\
&=&-M(t,0)^{-1}-M(t,0)^{-1}\\
&&\hspace*{10pt}\times(Q_-(0)^{-1}-M(t,0)^{-1})^{-1}M(t,0)^{-1}\\
&=&-M(t,0)^{-1}-M(t,0)^{-1}
\Sigma_0^{1/2}\left(\Sigma_0+\frac12 I\right.\\
&&\hspace*{-70pt}\left.- \left(\Sigma_0^{1/2}
\Sigma_T\Sigma_0^{1/2}+\frac{1}{4}I\right)^{1/2}-\Sigma_0^{1/2}M(t,0)^{-1}
\Sigma_0^{1/2}\right)^{-1}\Sigma_0^{1/2}
M(t,0)^{-1}
\end{eqnarray*}
for $t>0$.
For $t\in(0,1]$, the expression in parenthesis
\[
\Sigma_0+\frac12 I- \left(\Sigma_0^{1/2}
\Sigma_T\Sigma_0^{1/2}+\frac{1}{4}I\right)^{1/2}-\Sigma_0^{1/2}M(t,0)^{-1}
\Sigma_0^{1/2}
\]
is clearly maximal when $t=T$. However, for $t=T$ when $M(T,0)=I$, this expression is seen to be
\[
\frac12 I- \left(\Sigma_0^{1/2}
\Sigma_T\Sigma_0^{1/2}+\frac{1}{4}I\right)^{1/2} <0.
\]
Therefore, the expression in parenthesis is never singular and we deduce that $Q_-(t)^{-1}$ remains bounded for all $t\in(0,T]$, i.e., $Q_-(t)$ remains non-singular. For $t=0$, $Q(0)^{-1}$ is seen to be finite from \eqref{eq:finite}.
The argument for $P_-(t)$ is similar. Regarding ii), it suffices to notice that
$0<Q_+(0)<I$ while $Q_+(T)=Q_+(0)-I<0$. The statement ii) follows by continuity of $Q_+(t)$, and similarly for $P_+(t)$.

We now revert back to the set of coordinates where the drift is not necessarily zero and where $N(T,0)$ may not be the identity.
We see that
    \begin{eqnarray*}
        Q_\pm(0) &=& N(T,0)^{1/2}(Q_\pm(0))_{\rm new}N(T,0)^{1/2}\\
        &=&N(T,0)^{1/2}\Sigma_{0,\rm new}^{1/2}\left(\Sigma_{0,\rm new}+\frac12 I\right.\\
        &&\left.\pm \left(\Sigma_{0,\rm new}^{1/2}
            \Sigma_{T,\rm new}\Sigma_{0,\rm new}^{1/2}
            +\frac{1}{4}I\right)^{1/2}\right)^{-1}\\
            &&\Sigma_{0,\rm new}^{1/2}N(T,0)^{1/2}
    \end{eqnarray*}
where $\Sigma_{0,\rm new},\Sigma_{T,\rm new}$ as in (\ref{sigma0new}-\ref{sigmatnew}), which for compactness of notation in the statement of the proposition we rename $S_0$ and $S_T$, respectively.
\end{proof}

\begin{remark}
We have numerically observed that the iteration
\[
\begin{array}{cl}
P(0)&\\
\downarrow&\\
P(T)&=\Phi(T,0)P(0)\Phi(T,0)^\prime + M(T,0)\\
\downarrow&\\
Q(T)&=(\Sigma_T^{-1}-P(T)^{-1})^{-1}\\
\downarrow&\\
Q(0)&=\Phi(0,T)(Q(T)+M(T,0))\Phi(0,T)^\prime\\
\downarrow&\\
P(0)&=(\Sigma_0^{-1}-Q(0)^{-1})^{-1}
\end{array}
\]
using (\ref{bndphi2}) , converges to $Q_-(0)$, $P_-(0)$, $Q_-(T)$, $P_-(T)$, starting from a generic choice for $Q(0)$. The choice with a ``$-$'' is the one that generates the Schr\"odinger bridge as explained below. It is interesting to compare this property with similar properties of iterations that lead to solutions of Schr\"odinger systems in \cite{PT2,GP}.
\end{remark}

\begin{remark}
Besides the expression in the proposition, another equivalent formula for $Q_\pm(0)$ is
    \begin{eqnarray*}
        Q_\pm(0)&=&\Sigma_0^{1/2}(\frac{1}{2}I+
\Sigma_0^{1/2}\Phi(T,0)^\prime M(T,0)^{-1}\Phi(T,0)\Sigma_0^{1/2}\pm\\
&&(\frac{1}{4}I+\Sigma_0^{1/2}\Phi(T,0)^\prime M(T,0)^{-1}\Sigma_TM(T,0)^{-1}\\
&&\Phi(T,0)\Sigma_0^{1/2}
)^{1/2})^{-1}\Sigma_0^{1/2}
    \end{eqnarray*}
\end{remark}

\begin{remark}Interestingly, the solution $\Pi_+(t)=Q_+(t)^{-1}$ of the Riccati equation (\ref{R1}) corresponding to the choice ``$+$'' in $Q_\pm$ has a {\em finite escape time}.
\end{remark}

We are now in a position to {state the full} solution to the Schr\"odinger Bridge Problem \ref{formalization}. 

\begin{thm}\label{prop:bridge}
Assuming that the pair $(A(t),B(t))$ is controllable and that $\Sigma_0,\Sigma_T>0$, Problem \ref{formalization} has a unique optimal solution
\begin{equation}\label{eq:optimalcontrol}
u^8(t)=-B(t)'Q_-(t)^{-1}x(t)
\end{equation}
where $Q_-(\cdot)$ (together with a corresponding matrix function $P_-(\cdot)$) solves to the pair of coupled Lyapunov differential equations in Proposition \ref{prop:schpair}.
\end{thm}

\begin{proof}
{Since Proposition \ref{prop:schpair} has established existence and uniqueness of nonsingular solutions  $(P_-(\cdot),Q_-(\cdot))$ to the system \eqref{coupledLyapunov}, the result now follows from Corollary \ref{corollary}.}
\end{proof}

Thus, the controlled process \eqref{optevolution2} with $\Pi(t)=Q_-(t)^{-1}$,
\begin{equation}\label{eq:bridge}
dx^*=(A(t)-B(t)B(t)^\prime Q_-(t)^{-1}) x^*(t)dt+Bdw(t)
\end{equation}
steers the beginning density $\rho_0$ to the final one, $\rho_T$, with the least cost. Alternatively,
it forms a {least-effort} bridge between the two given marginals. It turns out that this controlled stochastic differential equation specifies the random evolution which is closest to the prior in the sense of relative entropy among those with the two given marginal distributions. This will be explained next. 

\section{Minimum relative entropy interpretation of optimal control}\label{sec:minentropy}

As noted earlier, there is a close relationship between the theory of large deviations, maximum entropy problems for random evolutions and stochastic optimal control \cite{Follmer,wakolbinger,daipra,daiprapavon,PW}.
In particular, classical Schr\"odinger bridges can be interpreted as both, a solution to a stochastic optimal control problem as well as inducing a probability law on path space that is consistent with given marginals and that is the closest to the prior in the sense of relative entropy. In other words, in effect, they answer the question of what the most likely path distribution is after ``conditioning'' the stochastic evolution on the two end-point marginals. Below we show that
the same property holds for the present case
of general stochastic linear system, i.e., of possibly degenerate linear diffusions.

\newcommand{\cP}{\mathcal P}
\newcommand{\cX}{\mathcal X}
\newcommand{\cPt}{\tilde{\mathcal P}}
For the purposes of this section we denote by $\cX=C([0,T];\mR^n)$
the space of continuous, $n$-dimensional sample paths of a linear diffusion as in \eqref{prior} and by
$\cP(\cdot)$ the induced probability measure on $\cX$.
One can describe $\cP(\cdot)$ as a mixture of measures pinned at the two ends of the interval $[0,T]$, that is,
\[
\cP(\cdot)=\int\cP(\cdot \mid x(0)=x_0,\,x(T)=x_T)\cP_{0,T}(dx_0dx_T)
\]
where $\cP(\cdot \mid x(0)=x_0,\,x(T)=x_T)$ is the conditional probability and $\cP_{0,T}(\cdot)$ is the joint probability of $(x(0),x(T))$.
The two end-point joint measure $\cP_{0,T}(\cdot)$, which is gaussian, has a (zero-mean) probability density function $g_{S_{0,T}}(x_0,x_T)$ with covariance
\begin{equation}\label{eq:S}
S_{0,T}=\left[\begin{matrix}S_0&S_0\Phi(T,0)'\\
\Phi(T,0)S_0&S_T\end{matrix}\right]
\end{equation}
where
\begin{align*}
S_0&=\E\{x_0x_0'\}\\
S_t&=\Phi(t,0)S_0\Phi(t,0)'+\int_0^t\Phi(t,\tau)B(\tau)B(\tau)'\Phi(t,\tau)'d\tau.
\end{align*}
In view of Sanov's theorem, see \cite[Section 3]{W}, Schr\"odinger's question reduces to identifying a probability law $\cPt(\cdot)$ on $\cX$ that minimizes the relative entropy
\newcommand{\cS}{{\mathcal S}}
\[
\cS(\cPt,\cP):=\int_\cX \log\left(\frac{d\cPt}{d\cP}\right)d\cPt
\]
among those that have the prescribed marginals. This is a very abstract problem.
However, if we  disintegrate $\cPt$
\[
\cPt(\cdot)=\int\cPt(\cdot \mid x(0)=x_0,\,x(T)=x_T)\cPt_{0,T}(dx_0dx_T),
\]
then the relative entropy can be readily written as the sum of two nonnegative terms,
the relative entropy between the two end-point joint measures
\[
\int \log\left(\frac{d\cPt_{0,T}}{d\cP_{0,T}}\right)\cPt_{0,T}
\]
plus
\[
\int \log\left(\frac{d\cPt(\cdot \mid x(0)=x_0,\,x(T)=x_T)}{d\cP(\cdot \mid x(0)=x_0,\,x(T)=x_T)}\right)\cPt.
\]
The second term becomes zero (and therefore minimal) when the conditional probability $\cPt(\cdot \mid x(0)=x_0,\,x(T)=x_T)$ is taken to be the same as
 $\cP(\cdot \mid x(0)=x_0,\,x(T)=x_T)$. Thus, the solution is in the same {\em reciprocal class} \cite{LK} as the prior evolution and, as already observed by Schr\"{o}dinger, the problem reduces to the finite-dimensional problem of minimizing relative entropy of the joint initial-final distribution among those that have the prescribed marginals.

It turns out that the probability law induced by \eqref{eq:bridge} is closest, in the relative entropy sense to the law of \eqref{prior}, that agrees with the two end-point marginal distributions at $t=0$ and $t=T$. Below we show this by verifying directly that  the densities between the two are identical when conditioned at the two end points, i.e., they share the same bridges, and that the end-point joint marginal for \eqref{eq:bridge} is indeed closest to the corresponding joint marginal for \eqref{prior}.

In order to show that two linear systems share the same bridges, we need the following lemma which is based on \cite{CG2014}.
\begin{lemma}\label{lemma:pin}
The probability law of the SDE \eqref{prior}, when conditioned on $x(0)=x_0,\,x(T)=x_T$, for any $x_0,\,x_T$, reduces to the probability law induced by the SDE
    \[
        dx=(A-BB'R(t)^{-1})xdt+BB'R(t)^{-1}\Phi(t,T)x_Tdt+Bdw
    \]
where $R(t)$ satisfies
    \[
        \dot{R}(t)=AR(t)+R(t)A'-BB'
    \]
with $R(T)=0$.
\end{lemma}

The stochastic process specified by this conditioning, or the latter SDE, will be referred to as the {\em pinned process associated to} \eqref{prior}.
Thus, in order to establish that the probability laws of~\eqref{eq:bridge} and~\eqref{prior} conditioned on $x(0)=x_0,\,x(T)=x_T$ are identical, it suffices to show that they have the same pinned processes for any $x_0,x_T$. This is done next.

\begin{thm}\label{thm:main2}
The probability law induced by \eqref{eq:bridge} represents the minimum of the relative entropy with respect to the law of \eqref{prior} over all probability laws on $\cX$ that have gaussian marginals with zero mean and covariances $\Sigma_0$ and $\Sigma_T$, respectively, at the two end-points of the interval $[0,T]$.
\end{thm}

\begin{proof}
We show that
i) the joint distribution
between the two end-points of $[0,T]$ for \eqref{eq:bridge} is the minimum of the relative entropy with respect to the corresponding two-end-point joint of \eqref{prior}, over distributions that satisfy the end-point constraint that the marginals are gaussian with specified covariances and, ii) the probability laws of these two SDEs on sample paths, conditioned on $x(0)=x_0,\,x(T)=x_T$ for any $x_0,\,x_T$ are identical by showing that they have the same pinned processes. We use the notation
\[
g_{S}(x):=(2\pi)^{-n/2}\det (S)^{-1/2}\exp\left[-\frac{1}{2}x'S^{-1}x\right],
\]
to denote
the standard Gaussian probability density function with mean zero and covariance $S$.

We start with i). In general, the relative entropy between two gaussian distributions $g_S(x)$ and $g_\Sigma(x)$ is
\begin{align}\nonumber
\int_{\mR^n} g_\Sigma(x)\log\left(\frac{g_\Sigma}{g_S}\right)dx=&\int_{\mR^n} g_\Sigma \log\left(\frac{\det(S)^{1/2}}{\det(\Sigma)^{1/2}}\right)dx\\\nonumber
&+\frac12 \int_{\mR^n} g_\Sigma(x)(x'S^{-1}x-x\Sigma x)dx\\\nonumber
=&\frac12\log(\det(S))-\frac12\log(\det(\Sigma))\\
&+\frac12\tr(S^{-1}\Sigma)-\frac12\tr(I).
\label{eq:Sbetweengaussians}
\end{align}
If $p_\Sigma$ is a probability density function, not necessarily gaussian, having covariance $\Sigma$, then
\begin{align}\nonumber
\int_{\mR^n} p_\Sigma(x)\log\left(\frac{p_\Sigma}{g_S}\right)dx=&\int_{\mR^n} p_\Sigma(x)\log\left(\frac{p_\Sigma}{g_S}\frac {g_\Sigma} {g_\Sigma}\right)dx\\
&\hspace*{-2.2cm}=\int_{\mR^n} p_\Sigma(x)\log\left(\frac{p_\Sigma }{g_\Sigma}\right)dx +
\int_{\mR^n} p_\Sigma(x)\log\left(\frac{g_\Sigma}{g_S}\right)dx\label{eq:firstterm}
\end{align}
where we multiplied and divided by $g_\Sigma$ and then partitioned accordingly.
We observe that
\[
\int_{\mR^n} p_\Sigma(x)\log\left(\frac{g_\Sigma}{g_S}\right)dx=
\int_{\mR^n} g_\Sigma(x)\log\left(\frac{g_\Sigma}{g_S}\right)dx.
\]
since $\log\left(\frac{g_\Sigma}{g_S}\right)$ is a quadratic form in $x$. Thus, the minimizer of relative entropy to
$g_S$ among probability density functions with covariance $\Sigma$ is gaussian since the first term in \eqref{eq:firstterm} is positive unless $p_\Sigma=g_\Sigma$, in which case it is zero.

We consider two-point joint gaussian distributions with covariances
$S_{0,T}$ as in \eqref{eq:S} with $S_0=\Sigma_0$, and
\[
\Sigma_{0,T}:=\left[\begin{matrix}\Sigma_0& Y'\\Y& \Sigma_T\end{matrix}\right]
\]
and evaluate $Y$ that minimizes the relative entropy. To this end we focus on
\begin{equation}\label{eq:entropy}
\tr(S_{0,T}^{-1}\Sigma_{0,T})-\log\det(\Sigma_{0,T}).
\end{equation}
Since
\[
S_{0,T}=\left[\begin{matrix}I\\\Phi(T,0)\end{matrix}\right]\Sigma_0
\left[\begin{matrix}I,&\Phi(T,0)'\end{matrix}\right]+
\left[\begin{matrix}0&0\\0& M(T,0)\end{matrix}\right],
\]
it follows that
\[
S_{0,T}^{-1}=
\left[\begin{matrix}\Sigma_0^{-1}+
\Phi' M^{-1} \Phi &-\Phi' M^{-1}\\-M^{-1}\Phi& M^{-1}\end{matrix}\right],
\]
where we simplified notation by setting
$\Phi:=\Phi(T,0)$ and $M:=M(T,0)$.
Then, the expression in \eqref{eq:entropy} becomes
\begin{align*}
&\tr\left((\Sigma_0^{-1}+\Phi'M^{-1}\Phi) \Sigma_0-\Phi'M^{-1}Y-Y'M^{-1}\Phi+M^{-1}\Sigma_T\right)\\
& -
\log\det(\Sigma_0)-\log\det(\Sigma_T-Y\Sigma_0^{-1}Y').
\end{align*}
Retaining only the terms that involve $Y$ leads us to seek a maximizing choice for $Y$ in
\[
f(Y):=\log\det(\Sigma_T-Y\Sigma_0^{-1}Y')+2\tr(\Phi'M^{-1}Y).
\]
Equating the differential of this last expression as a function of $Y$ to zero gives
    \begin{equation}\label{eq:firstorder}
        -2\Sigma_0^{-1}Y'(\Sigma_T-Y\Sigma_0^{-1}Y')^{-1}+2\Phi'M^{-1}=0
    \end{equation}
To see this, denote by $\Delta$ a small perturbation of $Y$ and retain the linear terms in $\Delta$ in
    \begin{eqnarray*}
&&f(Y+\Delta)-f(Y)\\
        &=& \log\det(I-(\Sigma_T-Y\Sigma_0^{-1}Y')^{-1}(\Delta\Sigma_0^{-1}Y'+Y\Sigma_0^{-1}\Delta'))\\&&+2\tr(\Phi'M^{-1}\Delta)\\
        &\simeq& -\tr((\Sigma_T-Y\Sigma_0^{-1}Y')^{-1}(\Delta\Sigma_0^{-1}Y'+Y\Sigma_0^{-1}\Delta'))\\&&+2\tr(\Phi'M^{-1}\Delta)\\
        &=& -2\tr(\Sigma_0^{-1}Y'(\Sigma_T-Y\Sigma_0^{-1}Y')^{-1}\Delta)\\&&+2\tr(\Phi'M^{-1}\Delta)
    \end{eqnarray*}

Let now
\[
        \Sigma_{0,T}=
        \left[\begin{matrix}
        \Sigma_0 & \Sigma_0\Phi_{Q_-}(T,0)'\\
        \Phi_{Q_-}(T,0)\Sigma_0 & \Sigma_T
        \end{matrix}\right]
    \]
where $\Phi_{Q_-}(T,0)$ is the state-transition matrix
of $A_{Q_-}(t)$, i.e., it satisfies
\begin{eqnarray*}
\frac{\partial}{\partial t}\Phi_{Q_-}(t,s)&=&A_{Q_-}(t)\Phi_{Q_-}(t,s), \mbox{ and}\\
-\frac{\partial}{\partial s}\Phi_{Q_-}(t,s)&=&\Phi_{Q_-}(t,s)A_{Q_-}(s),
\end{eqnarray*}
with $\Phi_{Q_-}(s,s)=I$. We need to show that $\Sigma_{0,T}$ here is the solution of the relative entropy minimization problem above. By concavity of $f(Y)$, it suffices to show that $Y=\Phi_{Q_-}(T,0)\Sigma_0$ satisfies the first-order condition \eqref{eq:firstorder}, that is,
    \begin{eqnarray*}
        &&\Phi_{Q_-}(T,0)'(\Sigma_T-\Phi_{Q_-}(T,0)\Sigma_0 \Phi_{Q_-}(T,0)')^{-1}\\
        &=& \Phi(T,0)'M(T,0)^{-1}\\
        &=& \Phi(T,0)'(S_T-\Phi(T,0)S_0\Phi(T,0)')^{-1},
    \end{eqnarray*}
where $S_t$ is as in \eqref{eq:S} with $S_0=\Sigma_0$. By taking inverse of both sides we obtain an equivalent formula
\begin{equation}\label{eq:cond2}
        \Sigma_T\Phi_{Q_-}(0,T)'-\Phi_{Q_-}(T,0)\Sigma_0=
        S_T\Phi(0,T)'-\Phi(T,0)\Sigma_0.
    \end{equation}
We claim
    \begin{eqnarray*}
        \Sigma_t\Phi_{Q_-}(0,t)'-\Phi_{Q_-}(t,0)\Sigma_0
        &=&
        S_t\Phi(0,t)'-\Phi(t,0)\Sigma_0,
    \end{eqnarray*}
then \eqref{eq:cond2} follows by taking $t=T$. We now prove our claim. For convenience, denote
    \begin{eqnarray*}
        F_1(t) &=& \Sigma_t\Phi_{Q_-}(0,t)'-\Phi_{Q_-}(t,0)\Sigma_0\\
        F_2(t) &=& S_t\Phi(0,t)'-\Phi(t,0)\Sigma_0\\
               F_3(t) &=& Q_-(t)(\Phi_{Q_-}(0,t)'-\Phi(0,t)').
    \end{eqnarray*}
 We will show that $F_1(t)=F_2(t)=F_3(t)$.
First we show $F_2(t)=F_3(t)$. Since $F_2(0)=F_3(0)=0$, we only need to show that they satisfy the same differential equation. To this end, compare
    \begin{eqnarray*}
        \dot{F}_2(t) &=& \dot{S}_t\Phi(0,t)'-S_t A'\Phi(0,t)'-A\Phi(t,0)\Sigma_0\\
        &=&(AS_t+S_t A'+BB')\Phi(0,t)'-S_t A'\Phi(0,t)'\\
        &&-A\Phi(t,0)\Sigma_0\\
        &=&
        AF_2(t)+BB'\Phi(0,t)',
    \end{eqnarray*}
with
    \begin{eqnarray*}
        \dot{F}_3(t) &=&
        {\dot{Q}_-}(t)(\Phi_{Q_-}(0,t)'-\Phi(0,t)')\\
        &&+Q_-(t)(-A_{Q_-}(t)'\Phi_{Q_-}(0,t)'+A'\Phi(0,t)')\\
        &=& (AQ_-(t)+Q_-(t)A'-BB')(\Phi_{Q_-}(0,t)'-\Phi(0,t)')\\
        &&-Q_-(t)A'(
        \Phi_{Q_-}(0,t)'-\Phi(0,t)')+BB'\Phi_{Q_-}(0,t)'\\
        &=& AF_3(t)+BB'\Phi(0,t)'
    \end{eqnarray*}
    which proves the claim $F_2(t)=F_3(t)$.
We next show that $F_1(t)=F_3(t)$. Let
    \begin{eqnarray*}
        H(t)&=&Q_-(t)^{-1}(F_3(t)-F_1(t))\\
        &=& -(Q_-(t)^{-1}-\Sigma_t^{-1})\Sigma_t\Phi_{Q_-}(0,t)'\\
        &&+Q_-(t)^{-1}\Phi_{Q_-}(t,0)\Sigma_0-\Phi(0,t)'\\
        &=& P(t)^{-1}\Sigma_t\Phi_{Q_-}(0,t)'\\
        &&+Q_-(t)^{-1}\Phi_{Q_-}(t,0)\Sigma_0-\Phi(0,t)',
    \end{eqnarray*}
then
\begin{eqnarray*}
        \dot{H}(t)&=&\dot{P}(t)^{-1}\Sigma_t\Phi_{Q_-}(0,t)'+P(t)^{-1}\dot{\Sigma}_t\Phi_{Q_-}(0,t)'
        -\\
        &&P(t)^{-1}\Sigma_t{A_{Q_-}(t)}'\Phi_{Q_-}(0,t)'
        +\dot{Q}_-(t)^{-1}\Phi_{Q_-}(t,0)\Sigma_0\\&&+
        Q_-(t)^{-1}A_{Q_-}(t)\Phi_{Q_-}(t,0)\Sigma_0+A'\Phi(0,t)'\\
        &=& -A'H(t).
    \end{eqnarray*}
Since $H(0)=Q_-(0)^{-1}(F_3(0)-F_1(0))=0$, it follows that $H(t)=0$ for all $t$, and therefore, $F_1(t)=F_3(t)$. This completes the proof of the first part.

We now prove ii). According to Lemma \ref{lemma:pin},
the pinned process corresponding to~\eqref{prior} satisfies
    \begin{equation}\label{eq:bridge11}
        dx=(A-BB'R_1(t)^{-1})xdt+BB'R_1(t)^{-1}\Phi(t,T)x_Tdt+Bdw
    \end{equation}
where $R_1(t)$ satisfies
    \[
        \dot{R_1}(t)=AR_1(t)+R_1(t)A'-BB'
    \]
with $R_1(T)=0$, while the pinned process corresponding to~\eqref{eq:bridge} satisfies
    \begin{eqnarray}\nonumber
        dx&=&(A_{Q_-}(t)-BB'R_2(t)^{-1})xdt+\\\label{eq:bridge22}
        &&BB'R_2(t)^{-1}\Phi_{Q_-}(t,T)x_Tdt+Bdw
    \end{eqnarray}
where $R_2(t)$ satisfies
    \[
        \dot{R_2}(t)=A_{Q_-}(t)R_2(t)+R_2(t)A_{Q_-}(t)'-BB'
    \]
with $R_2(T)=0$. We next show~\eqref{eq:bridge11} and~\eqref{eq:bridge22} are identical. It suffices to prove that
    \begin{equation}\label{eq:bridgecond11}
        A-BB'R_1(t)^{-1}=A_{Q_-}(t)-BB'R_2(t)^{-1}
    \end{equation}
and
    \begin{equation}\label{eq:bridgecond22}
        R_1(t)^{-1}\Phi(t,T)=R_2(t)^{-1}\Phi_{Q_-}(t,T).
    \end{equation}
Equation~\eqref{eq:bridgecond11} is equivalent to
    \[
        R_1(t)^{-1}=R_2(t)^{-1}+Q_-(t)^{-1}.
    \]
Multiply $R_1(t)$ and $R_2(t)$ on both sides to obtain
    \[
        R_2(t)=R_1(t)+R_1(t)Q_-(t)^{-1}R_2(t).
    \]
Now let
    \[
        J(t)=R_1(t)+R_1(t)Q_-(t)^{-1}R_2(t)-R_2(t).
    \]
Then
    \begin{eqnarray*}
        \dot{J}(t)&=&
        \dot{R_1}(t)+\dot{R_1}(t)Q_-(t)^{-1}R_2(t)+R_1(t)\dot{Q}_-(t)^{-1}R_2(t)\\
        &&+
        R_1(t)Q_-(t)^{-1}\dot{R_2}(t)-\dot{R_2}(t)\\
        &=& AJ+JA_{Q_-}(t)'.
    \end{eqnarray*}
Since
    \[
        J(T)=R_1(T)+R_1(T)Q_-(T)^{-1}R_2(T)-R_2(T)=0,
    \]
it follows that $J(t)=0$. This completes the proof of~\eqref{eq:bridgecond11}. Equation \eqref{eq:bridgecond22} is equivalent to
    \[
        \Phi(T,t)R_1(t)=\Phi_{Q_-}(T,t)R_2(t).
    \]
Let
    \[
        K(t)=\Phi(T,t)R_1(t)-\Phi_{Q_-}(T,t)R_2(t),
    \]
and then
    \begin{eqnarray*}
        \dot{K}(t)&=&
            -\Phi(T,t)AR_1(t)+\Phi(T,t)\dot{R}_1(t)+\\
            &&\Phi_{Q_-}(T,t)A_{Q_-}(t)R_2(t)-\Phi_{Q_-}(T,t)\dot{R}_2(t)\\
            &=& K(t)(A'-R_1(t)^{-1}BB').
    \end{eqnarray*}
Since
    \[
        K(T)=\Phi(T,T)R_1(T)-\Phi_{Q_-}(T,T)R_2(T)=0,
    \]
it follows that $K(t)=0$ as well for all $t$. This completes the proof.
\end{proof}

\section{Illustrative examples}\label{sec:examples}

We present two examples that illustrate the effect of optimal probability density steering. The first is based on inertial particles experiencing random accelerations and the second on an electrical circuit experiencing Nyquist-Johnson thermal noise from a resistor.

\subsection{Inertial particles}
Consider inertial particles experiencing random acceleration according to the model
    \begin{eqnarray*}
        dx(t)&=&v(t)dt\\
        dv(t)&=&u(t)dt+dw(t)
    \end{eqnarray*}
where $u(t)$ is a control force at our disposal, $x(t)$ represents position and $v(t)$ represents velocity. We want to squeeze the spread of the particles from an initial Gaussian distribution with $\Sigma_0=I$ at $t=0$ to a terminal marginal $\Sigma=\frac{1}{4}I$ at $t=1$. Figure \ref{fig:Eg1Phase1} shows sample paths in the phase space of $(x,v)$ as a function of time using the optimal stragegy for feedback control as explained earlier. Figure \ref{fig:Eg1Control1} displays the corresponding control action for each trajectory.

We provide two additional situations where the final distribution is localized in space and in velocity, respectively. The limit may be thought to approximate singular marginals, in each case, and it is of interest to compare the two since in one case the stochastic excitation affects directly the component of interest (velocity) whereas in the other after integration. Thus, we again take $\Sigma_0=I$ while we take $\Sigma_1$ to equal to ${\rm diag}(.05,\;1)$ and ${\rm diag}(1,\;.05)$, respectively, for the two cases. Sample paths in phase space under the optimal control law are shown in Figures \ref{fig:Eg1Phase2} and \ref{fig:Eg1Phase3}, respectively. In all of these phase plots \ref{fig:Eg1Phase1},\ref{fig:Eg1Phase2} and \ref{fig:Eg1Phase3}, the transparent blue ``tube'' represents the ``$3\,\sigma$'' tolerance interval. More specifically, the intersection ellipsoid between the tube and the slice plane $t$ is the set
    \[
        \left[\begin{matrix}x &v\end{matrix}\right]\Sigma_t^{-1}\left[\begin{matrix}x\\v\end{matrix}\right]\le 3^2.
    \]
\begin{figure}\begin{center}
    \includegraphics[width=0.47\textwidth]{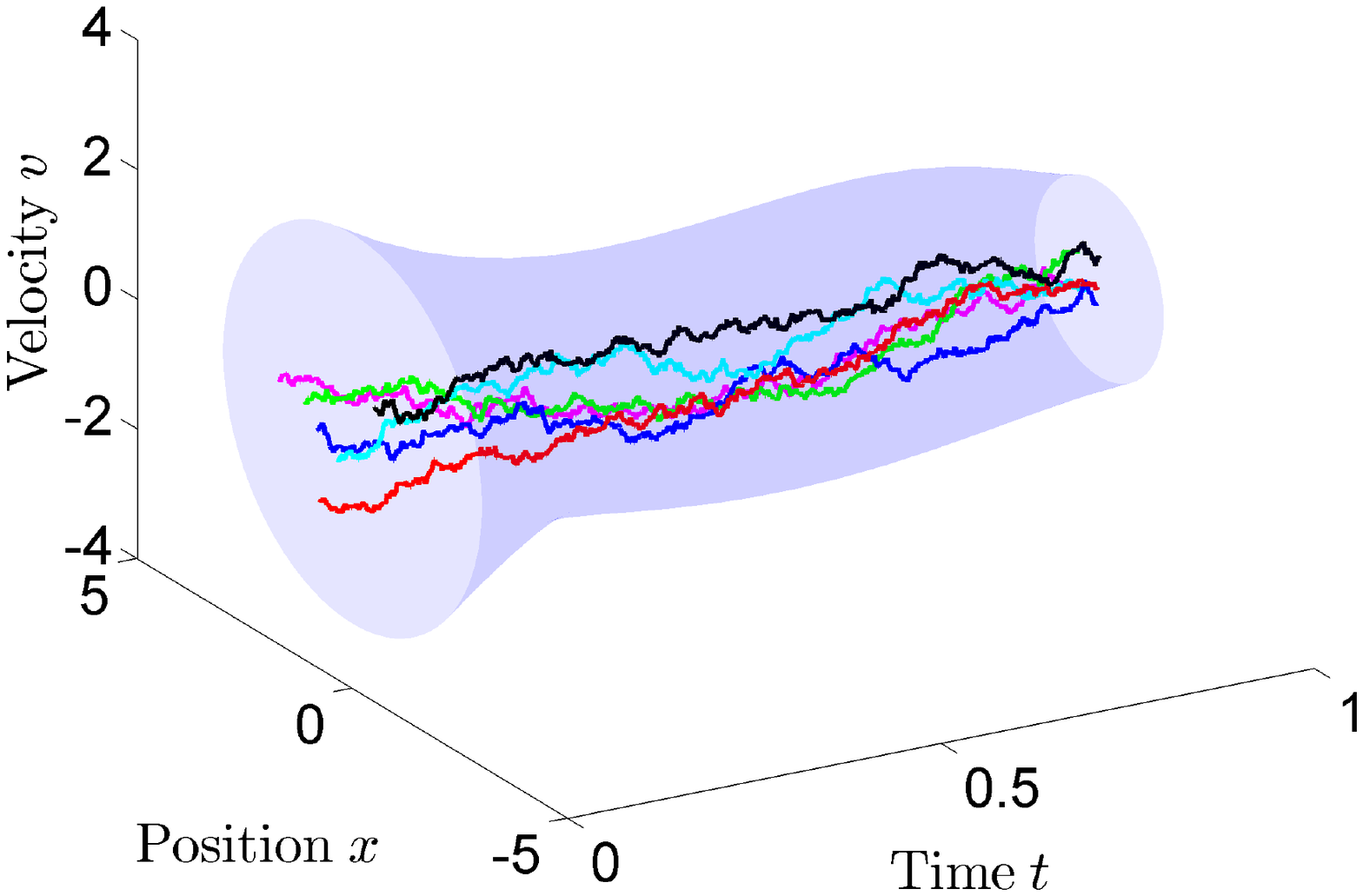}
    \caption{Inertial particles: state trajectories for $\Sigma_1=\frac{1}{4}I$}
    \label{fig:Eg1Phase1}
\end{center}\end{figure}
\begin{figure}\begin{center}
    \includegraphics[width=0.47\textwidth]{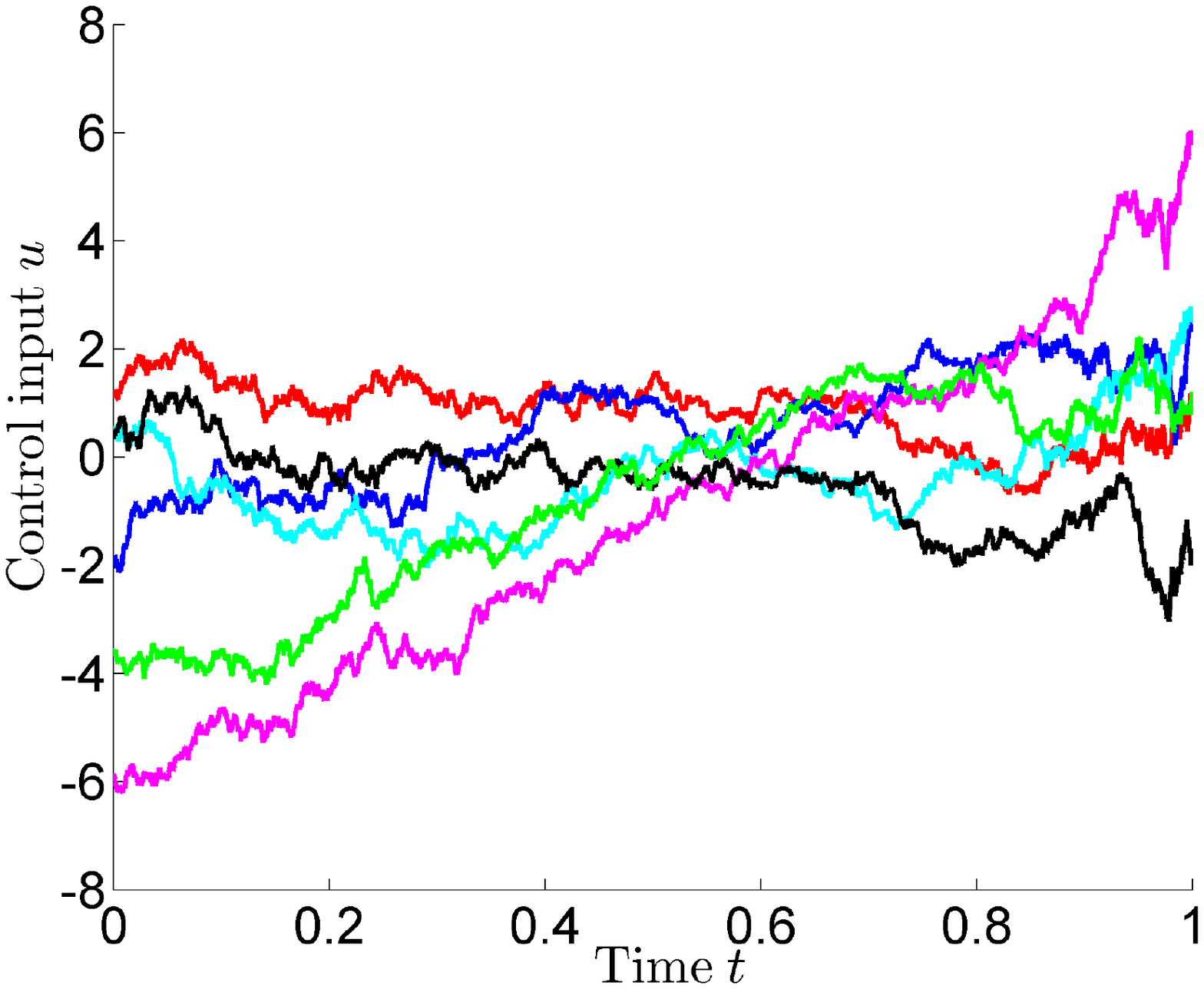}
    \caption{Inertial particles: control inputs for $\Sigma_1=\frac{1}{4}I$}
    \label{fig:Eg1Control1}
\end{center}\end{figure}
\begin{figure}\begin{center}
    \includegraphics[width=0.47\textwidth]{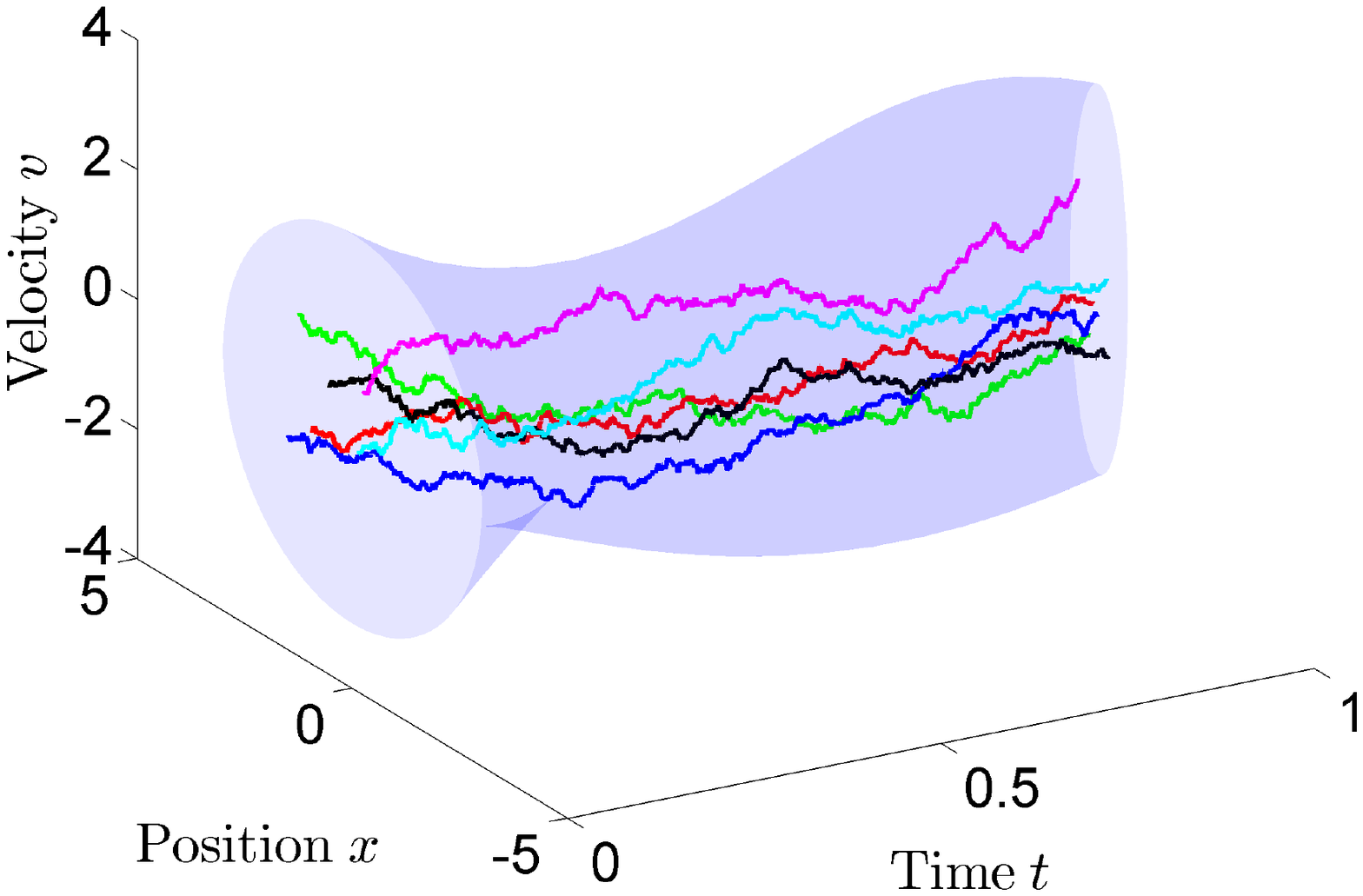}
    \caption{Inertial particles: state trajectories for $\Sigma_1={\rm diag}(.05,\;1)$}
    \label{fig:Eg1Phase2}
\end{center}\end{figure}
\begin{figure}\begin{center}
    \includegraphics[width=0.47\textwidth]{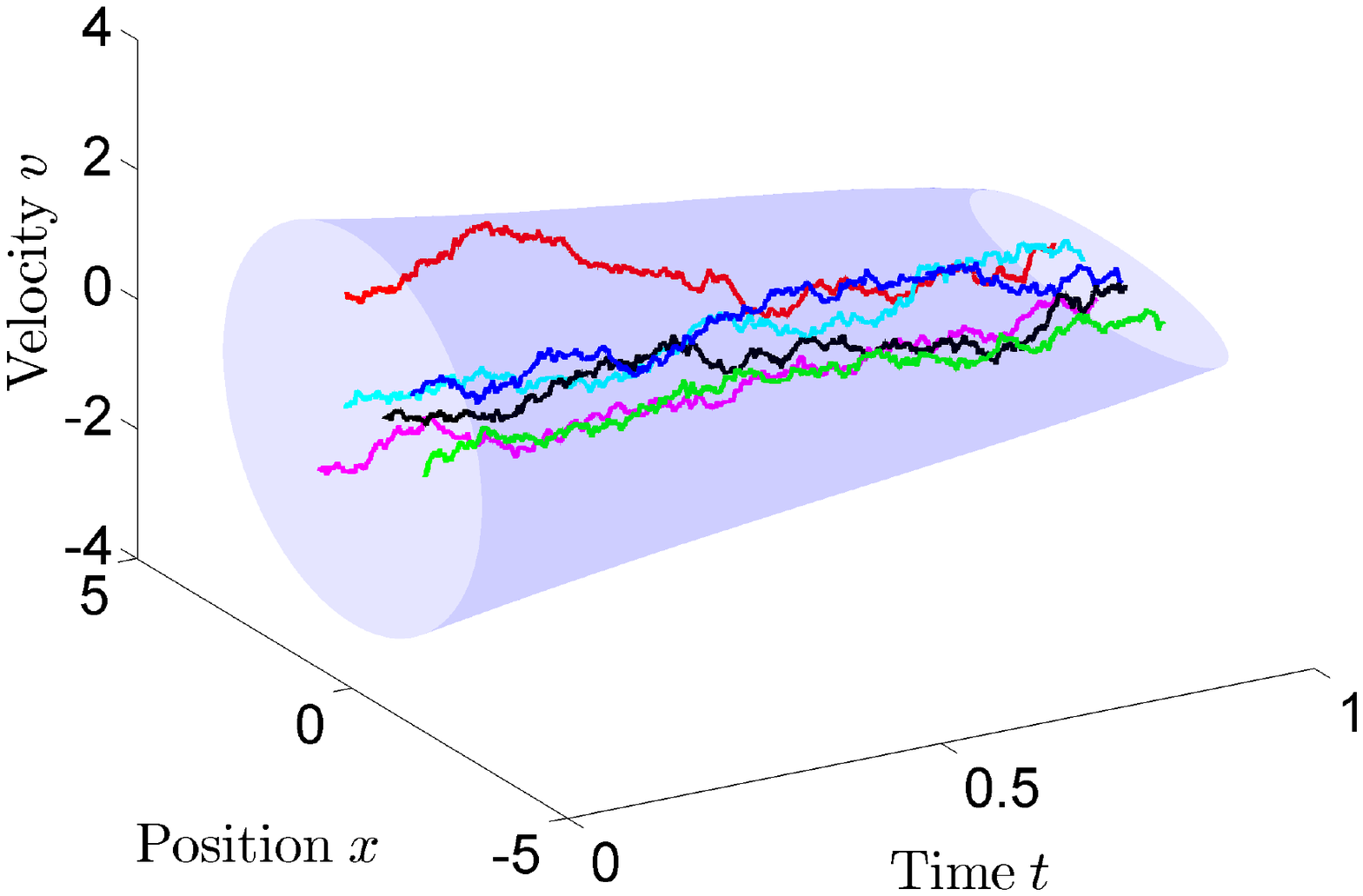}
    \caption{Inertial particles: state trajectories for $\Sigma_1={\rm diag}(1,\;.05)$}
    \label{fig:Eg1Phase3}
\end{center}\end{figure}

\subsection{Nyquist-Johnson resistor noise}

Consider the circuit in Figure \ref{fig:circuit} that includes  a resistor with a Nyquist-Johnson thermal noise voltage source. A model for the circuit is
    \begin{eqnarray*}
        Ldi_L(t)&=&\phantom{-}v_C(t)dt\\
        RCdv_C(t)&=&-v_C(t)dt-Ri_L(t)dt+u(t)dt+dw(t)
    \end{eqnarray*}
with all parameters $R=L=C=1$ in compatible units. Without any active control, i.e., when $u(t)\equiv 0$, the RLC circuit is driven by the thermal noise and reaches a steady state where the covariance matrix of the state vector $(i_L,\,v_C)'$ is $\frac{1}{2}I$.
Thus, we begin with random initial conditions for the state having an initial Gaussian distribution with $\Sigma_0=\frac{1}{2}I$ at $t=0$. Our aim is to specify the control voltage input $u(t)$ so as to reduce the effect of the thermal noise on the oscillator. As before, our target covariance at the end of a pre-specified interval $[0,\,1]$ is set to
to a terminal value; here this is $\Sigma_1=\frac{1}{16}I$. Figure \ref{fig:Eg3Phase1} shows the evolution of $(i_L,v_C)$ as a function of time under the effect of the least energy regulating input voltage $u(t)$ that aims to actively ``cool'' the resonator to its target final distribution. As before, Figure \ref{fig:Eg3Control1} displays the corresponding control inputs.
Once again, in \ref{fig:Eg3Phase1}, the transparent blue ``tube'' represents the ``$3\,\sigma$'' tolerance interval.

\begin{figure}\begin{center}
 \includegraphics[width=.42\textwidth]{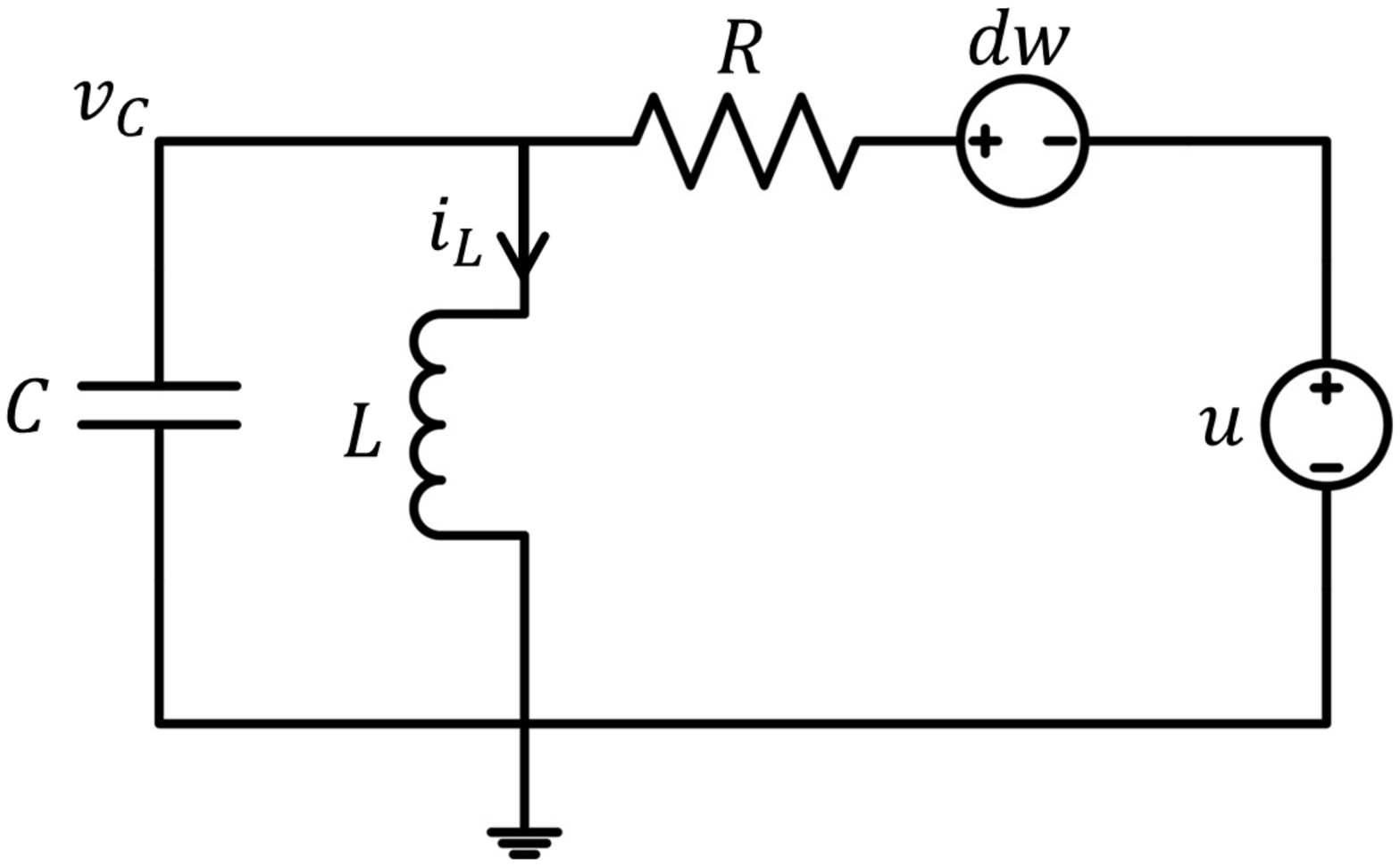}
    \caption{Circuit}
    \label{fig:circuit}
\end{center}\end{figure}

\begin{figure}\begin{center}
    \includegraphics[width=0.47\textwidth]{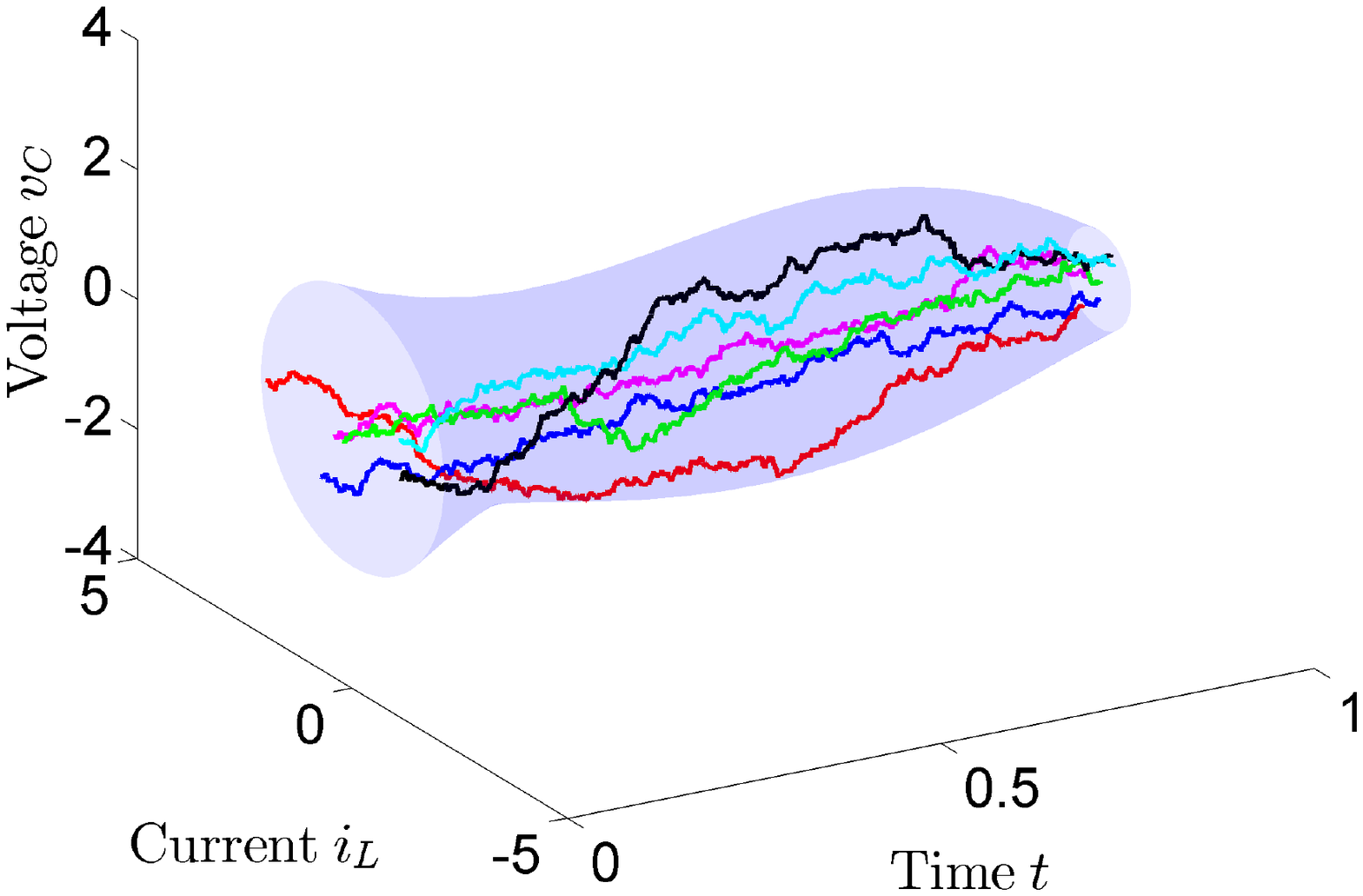}
    \caption{Nyquist-Johnson noise: trajectories for $\Sigma_1=\frac{1}{16}I$}
    \label{fig:Eg3Phase1}
\end{center}\end{figure}
\begin{figure}\begin{center}
    \includegraphics[width=0.47\textwidth]{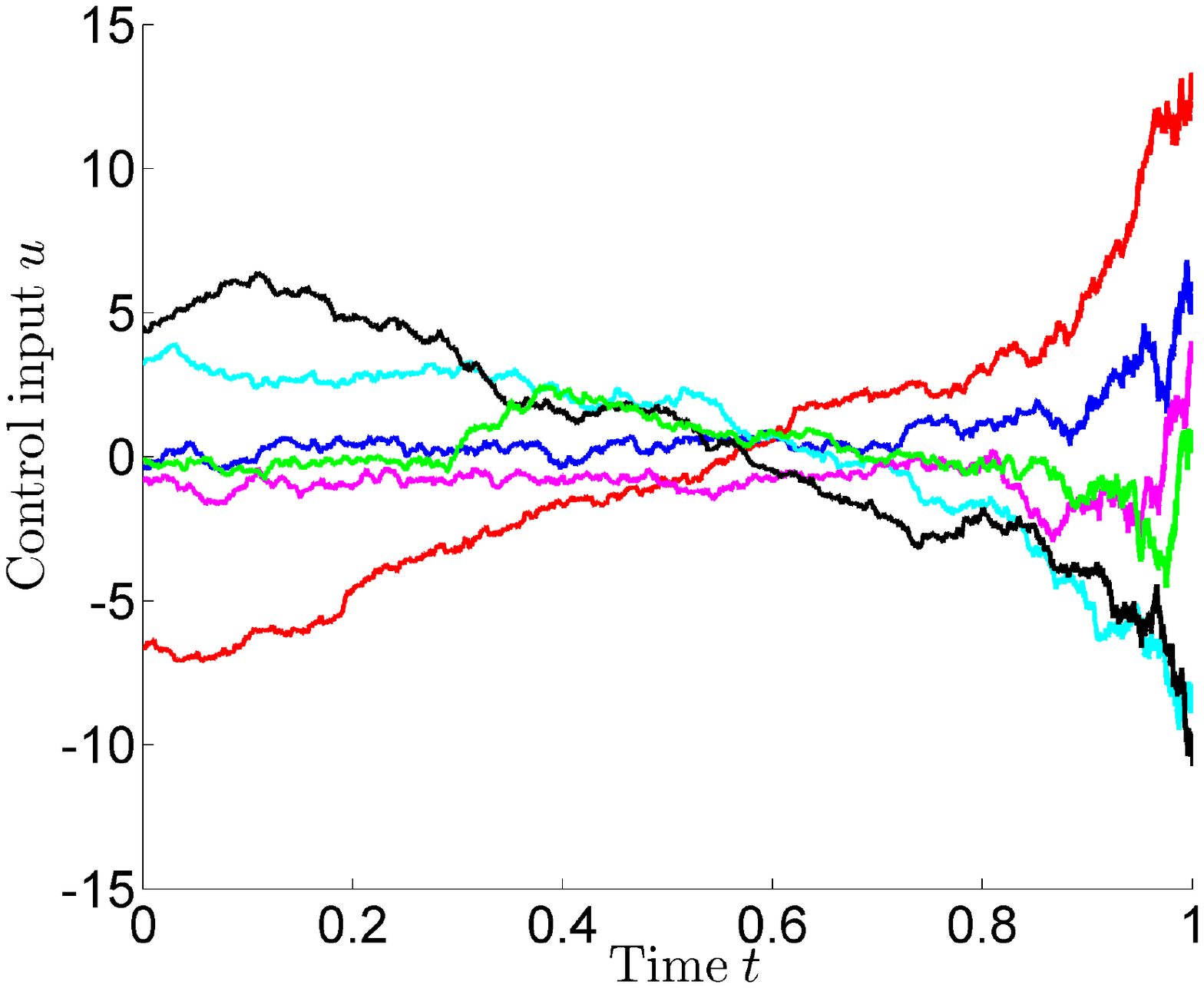}
    \caption{Nyquist-Johnson noise: controls for $\Sigma_1=\frac{1}{16}I$}
    \label{fig:Eg3Control1}
\end{center}\end{figure}

\section{Concluding remarks}
The problem to steer linear stochastic systems from a starting probability gaussian density to a target one with minimum effort has an explicit solution in feedback form. The minimum-energy control is computed by solving a pair of Lyapunov equations which are coupled through their boundary values at the two end-points of the interval. The stochastic process that is realized with the optimal control in place turns out to coincide with a solution to a seemingly different problem, that of seeking the most likely random evolution that connects the two marginals given a prior law in the form of the uncontrolled diffusion.
Both of these properties, the minimum energy and minimum relative entropy distance to the prior, generalize corresponding properties of classical Schr\'odinger bridges for nondegenerate diffusions.

The control of final distributions for stochastic systems and, in particular, the explicit form of solution in the present setting appears quite attractive for applications of active damping of nanomechanical systems and the ``cooling'' of stochastic thermal fluctuations.

\spacingset{1}

\bibliographystyle{IEEEtran}
\bibliography{refs}
\end{document}